\documentclass[12pt]{article}
\usepackage{amsmath, amssymb} 
\usepackage{amsthm} 
\usepackage{enumerate} 
\usepackage{url} 
\usepackage[utf8]{inputenc}
\usepackage{graphicx}
\usepackage{hyperref}
\usepackage{tikz}

\def\emph{\textsl}

\newtheoremstyle{plainsl}%
	{\topsep}
	{\topsep}
	{\slshape} 
	{}
	{\normalfont\bfseries}
	{.}
	{ }
	{}

{\theoremstyle{plainsl}
\newtheorem{theorem}{Theorem}[section]
\newtheorem{lemma}[theorem]{Lemma}
\newtheorem{remark}[theorem]{Remark}
\newtheorem{corollary}[theorem]{Corollary}
}
{\theoremstyle{remark}

}

\newcommand\fref[1]{Figure~\ref{fig:#1}}
\newcommand\lref[1]{Lemma~\ref{lem:#1}}
\newcommand\tref[1]{Theorem~\ref{thm:#1}}
\newcommand\cref[1]{Corollary~\ref{cor:#1}}

\def\G{\Gamma}
\def\F{\Gamma'}
\def\t{{^{\!\top\!\!}}}
\def\1{{\mathbf 1}}
\def\u{\mathbf{u}}
\def\v{\mathbf{v}}

\DeclareMathOperator\tr{tr}

\newcommand\cx{{\mathbb C}}
\newcommand\fld{{\mathbb F}}

\newcommand\ch[2]{\varPhi(#1,#2)}
\newcommand\gch[3]{\varPhi_{#3}(#1,#2)}
\newcommand\comp[1]{{\mkern2mu\overline{\mkern-2mu#1}}}

\newcommand{\eeq}[3]{#1 \equiv_{#3} #2}
\newcommand{\neeq}[3]{#1 \not\equiv_{#3} #2}
\newcommand{\congr}[3]{#1 \equiv #2 \ \mathopen{(}\textnormal{mod} \ #3\mathclose{)}}

\newcommand\by{\times}
\renewcommand\span{\mathop{\mathrm{span}}}

\newcommand\cD{{\mathcal D}}

\newcommand\cI{{\mathcal I}}

\begin{document}
\title{Descriptive Complexity of the\\ Generalized Spectra of Graphs}

\author{
Aida Abiad
\thanks{\texttt{a.abiad.monge@tue.nl},  Department of Mathematics and Computer Science, Eindhoven University of Technology, The Netherlands}
\thanks{Department of Mathematics: Analysis, Logic and Discrete Mathematics, Ghent University, Belgium} \thanks{ Department of Mathematics and Data Science, Vrije Universiteit Brussel, Belgium} 
\and
Anuj Dawar 
\thanks{\texttt{anuj.dawar@cl.cam.ac.uk}, Department of Computer Science and Technology, University of Cambridge, UK}
\and 
Octavio Zapata
\thanks{\texttt{octavioz@ciencias.unam.mx},  Departamento de Matemáticas, Facultad de Ciencias, Universidad Nacional Autónoma de México, México}
} 
\date{}
\maketitle

\begin{abstract}
Two graphs are cospectral if their respective adjacency matrices have the same multiset of eigenvalues, and generalized cospectral if they are cospectral and so
are their complements. We study generalized cospectrality in relation to logical definability. We show that any pair of graphs that are elementary equivalent with respect to the three-variable counting first-order logic $C^3$ are generalized cospectral, and this is not the case with $C^2$, nor with any number of variables if we exclude counting quantifiers. Using this result we provide a new characterization of the well-known class of distance-regular graphs using the logic $C^3$. 
We also show that, for controllable graphs (it is known that almost all graphs are controllable), the elementary equivalence in $C^2$ coincides with isomorphism. 
\end{abstract}

\section{Introduction} 
\label{sec:intro}
Let $\G$ be a graph on $n$ vertices with adjacency matrix $A$. 
The \emph{characteristic polynomial} $\ch{\G}{x}$ of $\G$ is the characteristic polynomial of $A$:
\[
    \ch{\G}{x} := \det(xI - A).
\]
Two graphs are \emph{cospectral} if they have the same characteristic polynomial. 
Since the adjacency matrices of isomorphic graphs are permutation similar,  we see that isomorphic graphs are cospectral. 
For any real number $y$ 
the \emph{generalized characteristic polynomial} $\gch{\G}{x}{y}$ of the adjacency matrix $A$ is defined by
\[
    \gch{\G}{x}{y} :=  \det(xI - yJ - A),
\]
where $J$ and $I$ denote the all-one matrix and the identity matrix, respectively.  
Two graphs are \emph{generalized cospectral} (sometimes also called $\mathbb{R}$-\emph{cospectral}) if they have the same polynomial $\varPhi_y$ for all values of $y$, that is, if the have the same generalized spectrum.
Since $\gch{\G}{x}{y}$ is the characteristic polynomial of $(A + yJ)$, it follows that isomorphic graphs are generalized cospectral. 
Since
\[
    \ch{\G}{x} = \gch{\G}{x}{0}
\]
we see that generalized cospectral graphs are cospectral. 
If $\comp{\G}$ denotes the complement of $\G$, then 
\[
    \ch{\comp{\G}}{x} = (-1)^n\ \gch{\G}{-x-1}{-1}
\]
and hence generalized cospectral graphs  have cospectral complements. Johnson and Newman \cite{JN} showed that cospectral graphs with cospectral complements are generalized cospectral.

It is still an open question whether almost all graphs are characterized by their generalized characteristic polynomial. 
More precisely, a graph $\G$ is \emph{determined by its generalized spectrum} if every graph which is generalized cospectral with $\G$ is isomorphic to $\G$.
It has been conjectured that the proportion of graphs on $n$ vertices which are determined by their generalized spectra goes to 1 as $n$ tends to infinity. 
Wang et al.~\cite{W2013,WX,MLW,W2017} have a number of results supporting this conjecture. They gave sufficient conditions for a graph to be determined by its generalized spectrum, and the majority of their results are proven for a wide class of graphs, the so-called controllable graphs. In fact, it is known that almost all graphs are controllable \cite{o2016conjecture}.

In this paper we study these concepts from the perspective of logical definability. We prove that generalized cospectrality is implied by elementary equivalence with respect to the three-variable counting first-order logic $C^3$. 
We also show that this is not the case with $C^2$, nor with any number of variables if we exclude counting quantifiers.
Using this result we prove that  generalized cospectrality coincides with $C^3$-equivalence for a relevant class of regular graphs (including all distance-regular graphs). Finally we show that, for controllable graphs, $C^2$-equivalence coincides with isomorphism. The latter yields a new proof of the classical result of Immerman and Lander \cite{immerman1990describing} that almost all graphs are definable (up to isomorphism) in the logic $C^2$. Our work implies some of the results in the work of Dawar, Severini and Zapata \cite{DSZ} and presents them in a more general setup.

\section{Generalized Cospectrality}
\label{sec:Rcospectrality}
 
If $u$ and $v$ are vertices from a graph $\G$ with adjacency matrix $A$,  the number of walks in $\G$ from $u$ to $v$ with length $\ell$ is equal to $(A^{\ell})_{u,v}$. 
Thus the number of closed walks of length $\ell$ in $\G$ is $\tr(A^{\ell})$, and so 
\[
    \sum_{\ell \geq 0}\tr(A^\ell ) x^\ell
\]
is the generating function for closed walks in $\G$ counted by length. 
If $\varPhi^{\ \!\!\prime}(\G, x)$ denotes the derivative of $\ch{\G}{x}$ with respect to $x$, we find that 
\[
    \sum_{\ell \geq 0}\tr(A^\ell ) x^\ell = \frac{x^{-1}\ \varPhi^{\ \!\!\prime}(\G, x^{-1})}{\  \ch{\G}{x^{-1}}}.
\]
Therefore, two graphs are cospectral if and only if they have the same number of closed walks of each length.

If $\G$ and $\F$ are non-cospectral graphs, there exists some $\ell\geq 1$ such that the number of closed walks of length $\ell$ in $\G$ and $\F$ is not the same.  
This fact can be expressed in the language of first-order logic using counting quantifiers and not more than 3 variables. 
The language $L^k$ consists of the fragment of first-order logic in which only $k\geq 1$ distinct variables can be used.
We use $C^k$ to denote extension of $L^k$ with \emph{counting quantifiers}: for each non-negative integer $r$, we have a quantifier $\exists^{\geq r}$ whose semantics is defined so that $\exists^{\geq r} \phi(x)$ is true in a graph if there are at least $r$ distinct vertices which can be substituted for $x$ to make $\phi$ true. We use the abbreviation $\exists^{=r}\phi$ for the formula $\exists^{\geq r}\phi \land \lnot \exists^{\geq r+1}\phi$ that asserts the existence of exactly $r$ vertices satisfying $\phi$. 

Two graphs $\G$ and $\F$ are \emph{elementary equivalent} with respect to a first-order language $L$ (or \emph{$L$-equivalent}), just in case $\G \models \phi$ if and only if $\F \models \phi$ for any $L$-sentence $\phi$. 
In other words, $L$-equivalent graphs are precisely those graphs that cannot be distinguished by any property defined by a sentence in the language $L$. 
We write $\eeq{\G}{\F}{L}$ to denote that $\G$ and $\F$ are elementary equivalent with respect to $L$. 

The next result shows that $C^3$-equivalence implies generalized cospectrality. 

\begin{theorem} \label{thm:c3}
$C^3$-equivalent graphs are generalized cospectral. 
\end{theorem}
 
\begin{proof}[\emph{Proof}]
For every non-negative integer $r$, we use counting quantifiers and not more than three variables to write a formula $\psi_{\ell}^r(x,y)$ that asserts the existence of exactly $r$ distinct walks of length $\ell \geq 1$ between $x$ and $y$.
We define $\psi_{\ell}^r(x,y)$ by induction on $\ell$. 
For $\ell = 1$, let
\[
    \psi_{1}^{0}(x,y) := \lnot E(x,y), \quad 
    \psi_{1}^{1}(x,y) := E(x,y) \quad \textnormal{and}\quad  \psi_{1}^{r}(x,y) := \bot\ \textnormal{ if }\ r > 1.
\]
Assuming that we have defined $\psi_{\ell}^r(x,y)$ for every $r$, we proceed with the definitions for $\ell + 1$:  
\[
    \psi_{\ell + 1}^{0}(x,y) := \forall z ( E(x,z) \to \psi_{\ell}^{0}(z,y) ).
\]
An \emph{integer partition} of $r>1$ is a sequence $(r_1^{a_1}, r_2^{a_2}, \dots, r_d^{a_d})$ of $d\leq r$ pairwise distinct positive integers (called the parts of $r$), 
such that for any $i$ ($1 \leq i \leq d$) the part $r_i\geq 0$ is repeated $a_i \geq 1$ times and $r = \sum_{i=1}^{d} a_i r_i$.
Let $\Pi_r$ denote the set of all integer partitions of $r$ and note that this is a finite set. 
Then 
\[
    \psi_{\ell + 1}^{r}(x,y) := \bigvee_{(r_1^{a_1},\ldots, r_d^{a_d})\in \Pi_r} [ \bigwedge_{i = 1}^{d}  ( \exists^{= a_i} x\ \psi_{\ell}^{r_i}(z,y)) \land  \exists^{\geq a}z\  E(x,z)  ]
\]
where $a = \sum_{i=1}^{d}a_i$. 
Here $a_i$ denotes the number of neighbors of $x$ for which there are exactly $r_i$ walks of length $\ell$ from each of them to $y$.
 
All of these definitions make use of at most three variables, and therefore each $\psi_{\ell}^{r}$ is a formula of the language $C^3$.
Thus we can write a $C^3$-sentence $\phi_{\ell}^{r}$ which is true in a graph when the total number of  closed walks of length $\ell$ is exactly $r$:
\[
    \phi_{\ell}^{r} := \bigvee_{(r_1^{a_1},\dots, r_d^{a_d})\in \Pi_r} \bigwedge_{i = 1}^{d} \exists^{=a_i}x \exists y\ (\psi_{\ell}^{r_i}(x,y) \land x = y).  
\] 
 
If $\G$ and $\F$ are not generalized cospectral graphs, then either $\G$ and $\F$ are non-cospectral or $\comp{\G}$ and $\comp{\F}$ are non-cospectral. 
Without loss of generality we assume that $\G$ and $\F$ do not have the same number of closed walks of length $\ell$ for some $\ell\geq 1$. 
Let $r$ be the number of closed walks of length $\ell$ in $\G$.   
Then $\G \models \phi_{\ell}^{r}$ and $\F\not\models \phi_{\ell}^{r}$
and therefore $\neeq{\G}{\F}{C^3}$.  
\end{proof}

Theorem \ref{thm:c3} is optimal in the sense that $C^2$-equivalence is not sufficient for generalized cospectrality.
To see this, let $\G$ be the disjoint union of two triangles and let $\F$ be a cycle of length 6. 
A simple pebble game argument shows that $\eeq{\G}{\F}{C^2}$ (see e.g.~\cite[Proposition 4.7.4]{immerman1990describing}). 
However, $\G$ and $\F$ are not generalized cospectral because they have a different number of triangles and the spectrum determines the number of triangles.

The converse of \tref{c3} does not hold in general. 
Suppose, for instance, that $\G$ and $\F$ are the left and right graphs from \fref{smallestRcospectralpair}, respectively. This is the smallest pair of nonisomorphic generalized cospectral graphs. 
Since $\G$ has one isolated vertex and $\F$ is connected, then $\G \models \exists x \forall y \lnot E(x,y)$ and $\G \not\models \exists x \forall y \lnot E(x,y)$. 
Hence $\neeq{\G}{\F}{C^2}$, and therefore $\neeq{\G}{\F}{C^k}$ for $k\geq 2$.

\begin{figure}[ht!]
\centering
\begin{tikzpicture}[x=0.55pt,y=0.55pt,yscale=-1,xscale=1, thick]
\draw   (123,100.5) -- (100.25,139.9) -- (54.75,139.9) -- (32,100.5) -- (54.75,61.1) -- (100.25,61.1) -- cycle ;
\draw    (54.75,61.1) -- (100.25,61.1) ;
\draw [shift={(100.25,61.1)}, rotate = 0] [color={rgb, 255:red, 0; green, 0; blue, 0 }  ][fill={rgb, 255:red, 0; green, 0; blue, 0 }  ][line width=0.75]      (0, 0) circle [x radius= 3.35, y radius= 3.35]   ;
\draw [shift={(54.75,61.1)}, rotate = 0] [color={rgb, 255:red, 0; green, 0; blue, 0 }  ][fill={rgb, 255:red, 0; green, 0; blue, 0 }  ][line width=0.75]      (0, 0) circle [x radius= 3.35, y radius= 3.35]   ;
\draw    (123,100.5) -- (100.25,139.9) ;
\draw [shift={(100.25,139.9)}, rotate = 120] [color={rgb, 255:red, 0; green, 0; blue, 0 }  ][fill={rgb, 255:red, 0; green, 0; blue, 0 }  ][line width=0.75]      (0, 0) circle [x radius= 3.35, y radius= 3.35]   ;
\draw [shift={(123,100.5)}, rotate = 120] [color={rgb, 255:red, 0; green, 0; blue, 0 }  ][fill={rgb, 255:red, 0; green, 0; blue, 0 }  ][line width=0.75]      (0, 0) circle [x radius= 3.35, y radius= 3.35]   ;
\draw    (32,100.5) -- (54.75,139.9) ;
\draw [shift={(54.75,139.9)}, rotate = 60] [color={rgb, 255:red, 0; green, 0; blue, 0 }  ][fill={rgb, 255:red, 0; green, 0; blue, 0 }  ][line width=0.75]      (0, 0) circle [x radius= 3.35, y radius= 3.35]   ;
\draw [shift={(32,100.5)}, rotate = 60] [color={rgb, 255:red, 0; green, 0; blue, 0 }  ][fill={rgb, 255:red, 0; green, 0; blue, 0 }  ][line width=0.75]      (0, 0) circle [x radius= 3.35, y radius= 3.35]   ;
\draw  [dash pattern={on 0.75pt off 750pt}]  (77.5,100.5) -- (100.25,139.9) ;
\draw [shift={(100.25,139.9)}, rotate = 60] [color={rgb, 255:red, 0; green, 0; blue, 0 }  ][fill={rgb, 255:red, 0; green, 0; blue, 0 }  ][line width=0.75]      (0, 0) circle [x radius= 3.35, y radius= 3.35]   ;
\draw [shift={(77.5,100.5)}, rotate = 60] [color={rgb, 255:red, 0; green, 0; blue, 0 }  ][fill={rgb, 255:red, 0; green, 0; blue, 0 }  ][line width=0.75]      (0, 0) circle [x radius= 3.35, y radius= 3.35]   ;

\draw    (250,100.5) -- (272.75,139.9) ;
\draw [shift={(272.75,139.9)}, rotate = 60] [color={rgb, 255:red, 0; green, 0; blue, 0 }  ][fill={rgb, 255:red, 0; green, 0; blue, 0 }  ][line width=0.75]      (0, 0) circle [x radius= 3.35, y radius= 3.35]   ;
\draw [shift={(250,100.5)}, rotate = 60] [color={rgb, 255:red, 0; green, 0; blue, 0 }  ][fill={rgb, 255:red, 0; green, 0; blue, 0 }  ][line width=0.75]      (0, 0) circle [x radius= 3.35, y radius= 3.35]   ;
\draw    (250,100.5) -- (227.25,139.9) ;
\draw [shift={(227.25,139.9)}, rotate = 120] [color={rgb, 255:red, 0; green, 0; blue, 0 }  ][fill={rgb, 255:red, 0; green, 0; blue, 0 }  ][line width=0.75]      (0, 0) circle [x radius= 3.35, y radius= 3.35]   ;
\draw [shift={(250,100.5)}, rotate = 120] [color={rgb, 255:red, 0; green, 0; blue, 0 }  ][fill={rgb, 255:red, 0; green, 0; blue, 0 }  ][line width=0.75]      (0, 0) circle [x radius= 3.35, y radius= 3.35]   ;
\draw    (227.25,139.9) -- (204.5,179.31) ;
\draw [shift={(204.5,179.31)}, rotate = 120] [color={rgb, 255:red, 0; green, 0; blue, 0 }  ][fill={rgb, 255:red, 0; green, 0; blue, 0 }  ][line width=0.75]      (0, 0) circle [x radius= 3.35, y radius= 3.35]   ;
\draw [shift={(227.25,139.9)}, rotate = 120] [color={rgb, 255:red, 0; green, 0; blue, 0 }  ][fill={rgb, 255:red, 0; green, 0; blue, 0 }  ][line width=0.75]      (0, 0) circle [x radius= 3.35, y radius= 3.35]   ;
\draw    (272.75,139.9) -- (295.5,179.31) ;
\draw [shift={(295.5,179.31)}, rotate = 60] [color={rgb, 255:red, 0; green, 0; blue, 0 }  ][fill={rgb, 255:red, 0; green, 0; blue, 0 }  ][line width=0.75]      (0, 0) circle [x radius= 3.35, y radius= 3.35]   ;
\draw [shift={(272.75,139.9)}, rotate = 60] [color={rgb, 255:red, 0; green, 0; blue, 0 }  ][fill={rgb, 255:red, 0; green, 0; blue, 0 }  ][line width=0.75]      (0, 0) circle [x radius= 3.35, y radius= 3.35]   ;
\draw    (250,61.1) -- (250,100.5) ;
\draw [shift={(250,100.5)}, rotate = 90] [color={rgb, 255:red, 0; green, 0; blue, 0 }  ][fill={rgb, 255:red, 0; green, 0; blue, 0 }  ][line width=0.75]      (0, 0) circle [x radius= 3.35, y radius= 3.35]   ;
\draw [shift={(250,61.1)}, rotate = 90] [color={rgb, 255:red, 0; green, 0; blue, 0 }  ][fill={rgb, 255:red, 0; green, 0; blue, 0 }  ][line width=0.75]      (0, 0) circle [x radius= 3.35, y radius= 3.35]   ;
\draw    (250,21.69) -- (250,61.1) ;
\draw [shift={(250,61.1)}, rotate = 90] [color={rgb, 255:red, 0; green, 0; blue, 0 }  ][fill={rgb, 255:red, 0; green, 0; blue, 0 }  ][line width=0.75]      (0, 0) circle [x radius= 3.35, y radius= 3.35]   ;
\draw [shift={(250,21.69)}, rotate = 90] [color={rgb, 255:red, 0; green, 0; blue, 0 }  ][fill={rgb, 255:red, 0; green, 0; blue, 0 }  ][line width=0.75]      (0, 0) circle [x radius= 3.35, y radius= 3.35]   ;
\end{tikzpicture}
\caption{Smallest pair of generalized cospectral graphs with respect to the adjacency matrix.}\label{fig:smallestRcospectralpair}
\end{figure}
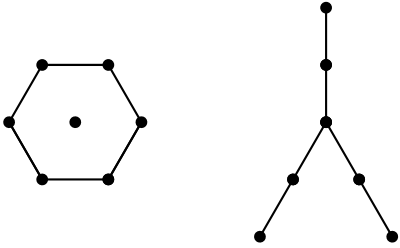

The use of counting is essential for the distinguishability of the generalized spectra in finite-variable logics.

\begin{theorem}\label{thm:nocounting}
There is a pair of nonisomorphic $L^k$-equivalent graphs which are not generalized cospectral for each positive integer $k$.  
\end{theorem}
\begin{proof}[\emph{Proof}]
For $k \geq 1$, let $n_k$ be the smallest integer greater than 
$2(k \log(4) + \log(k))/ \log(13)$ and let $q_k = 13^{n_k}$. 
Then $q_k > k^2 2^{4k}$. 
We write $\fld_{q_k}$ to denote the finite field of order $q_{k}$.

Let $\G_k$ and $\F_k$ be a pair of graphs with vertex set $\fld_{q_k}$ defined as follows. 
Two vertices $i$ and $j$ are adjacent in $\G_k$ if there exists a positive integer $x$ such that $\congr{x^2}{(i-j)}{q_k}$.
The vertices $i$ and $j$ are adjacent in $\F_k$ if there exists a positive integer $x$ such that $\congr{x^3}{(i-j)}{q_k}$. 
Hence, $\G_k$ is a Paley graph and $\F_k$ is a cubic Paley graph. 

It can be shown that with our choice of parameters $\eeq{\G_k}{\F_k}{L^k}$ (see \cite{blass1981paley} and also \cite{ananchuen2006cubic}). 
Since Paley graphs and cubic Paley graphs are strongly regular, the parameters (which include the degree) of these classes of graphs determine their characteristic polynomials. 
Thus $\G_k$ and $\F_k$ are cospectral if and only if they have  the same parameters. 
However, $\G_k$ is regular of degree $(q_k - 1)/2$ and $\F_k$ is regular of degree $(q_k - 1)/3$. 
Therefore $\G_k$ and $\F_k$ are not isomorphic and not generalized cospectral. 
\end{proof}

\section{Coherent Algebras}
\label{sec:coherent}
A matrix algebra is a real or complex vector space of matrices closed under matrix multiplication. 
The \emph{Schur product} of two $n\times n$ matrices $A$ and $B$ is the $n\times n$ matrix $A \circ B$ given by 
\[
    (A \circ B)_{i,j}  = (A)_{i,j}(B)_{i,j}.
\]
This is a commutative and associative operation with the identity $J$.
A \emph{coherent algebra} is a matrix algebra closed under conjugate-transpose and Schur multiplication which contains $I$ and $J$.  

If $A$ is the adjacency matrix of a graph $\G$, the \emph{adjacency algebra} of $\G$ is the matrix algebra $\cx[A]$ of polynomials in $A$ with complex coefficients. 
It is the unique minimal matrix algebra that contains $A$. 
The \emph{coherent closure} of $\G$ is the intersection of all matrix algebras that contain $\cx[A]$ which are closed under Schur product and conjugate-transpose. 
It is the unique minimal coherent algebra that contains $\cx[A]$.

A homomorphism of algebras is a ring homomorphism that commutes with scalar multiplication. 
An algebra homomorphism between two coherent algebras need not preserve the Schur product.
A \emph{coherent homomorphism} is an algebra homomorphism that commutes with Schur multiplication.

Any vector space of matrices which is closed under Schur product has a unique basis of matrices with entries 0 and 1, see for instance \cite[Theorem 2.6.1(i)]{brouwer1989}. 
An algebra homomorphism between two coherent algebras commutes with Schur product if and only if it maps the 01-basis of one algebra to the 01-basis of the other (see e.g. \cite[Proposition 2.3.17]{chen2019lectures}). 

Weisfeiler and Leman  \cite{weisfeiler1968reduction} provided an algorithm that computes in polynomial time the 01-basis of  the coherent closure of any graph. 
Two graphs are called \emph{WL-equivalent} if there is an invertible coherent homomorphism between their coherent closures.  
The adjacency matrix of a graph does not necessarily belong to the 01-basis of  its coherent closure. 
However, we have the following result which implies that WL-equivalent graphs are generalized cospectral. 
It follows from   \cite[Theorem 5.2]{cai1992optimal}. 

\begin{lemma} \label{lem:wl}
Two graphs are WL-equivalent  if and only if they are  $C^3$-equivalent. 
\end{lemma}

\section{Distance-Regular Graphs}
\label{sec:drg}

Let $\G$ be a graph on $n$ vertices with diameter $d$. 
We say that $\G$ is \emph{distance-regular} with intersection array $\{b_0,b_1,\dots, b_{d-1};c_1,c_2,\dots,c_d\}$ if $\G$ is regular of degree $k = b_0$, and if for any two vertices $u$ and $v$ at distance $i$, there are precisely $c_i$ neighbors of $v$ at distance $i - 1$ from $u$, and $b_i$ neighbors of $v$ at distance $i+1$ from $u$. 
It is well known that two distance-regular graphs are cospectral if and only if they have the same intersection array (see e.g. \cite{brouwer1989}). 
Since the only vertex at distance 0 from a given vertex is itself,  we have that $c_1 = 1$.
The five Platonic solids are simple examples: the tetrahedron has intersection array $\{3;1\}$, the octahedron $\{4,1; 1,4\}$, the cube $\{3,2,1; 1,2,3\}$, the icosahedron $\{5,2,1; 1,2,5\}$ and the dodecahedron $\{3,2,1,1,1; 1,1,1,2,3\}$.
Every strongly regular graph with parameters $\{n,k;a,c\}$ is a distance-regular graph with diameter at most 2 and  intersection array $\{k, k-a-1; 1, c\}$.

The $i$-th \emph{distance matrix} of $\G$ is the $n\by n$ 01-matrix $A_i$ with $uv$-entry equal to 1 if and only if the distance between vertex $u$ and vertex $v$ in $\G$ is $i$. 
By convention, we set $A_0 = I$ and $A_i = 0$ if $i$ is greater than $d$. 
Then $\sum_{i=0}^d A_i = J$ and so $\cD = \{A_0,A_1,\dots, A_d\}$ is a linearly independent set of symmetric 01-matrices. 
It follows that $\span(\cD)$ is a Schur-closed vector space of symmetric matrices which contains $I$ and $J$.

Since $A_1$ is the adjacency matrix of $\G$, if we set $a_i = k - b_{i} - c_{i}$ for $i= 0, \dots, d$ where $b_{-1}= b_d = c_0 =c_{d+1}=0$, then we have  
\[
    A_1A_i  =  b_{i-1} A_{i-1} + a_i A_i + c_{i+1} A_{i+1}.  
\]
If $i$ is less than $d$ then $c_{i+1}$ is non-zero, and thus $A_{i+1}$ is a linear combination of $A_{i-1}, A_i$ and $ A_1A_i $.
We see by induction that $A_i$ is a polynomial of degree $i$ in $A_1$ for $i = 1,\dots,d$. 
Again by induction, we find that $A_iA_j \in \span(\cD)$ for all $i$ and $j$ (see \cite[Lemma 4.2.2]{godsil1993algebraic}). 
Thus $\span(\cD)$ is closed under matrix multiplication and so it is a coherent algebra.

Consider two graphs $\G$ and $\F$ with adjacency matrices $A$ and $A'$, respectively. 
Since  $A$ and $A'$ are real symmetric matrices, they are diagonalizable. 
This implies that if $A$ and $A'$  have the same characteristic polynomial, then $Q^\t AQ=A'$ for some orthogonal matrix $Q$. 
It follows that cospectral graphs have similar adjacency matrices.  
In particular, if $\G$ and $\F$ are cospectral distance-regular graphs, then all their distance matrices are simultaneously similar. This is because the $i$-th distance matrix of $\G$ is a polynomial of degree $i$ in the adjacency matrix of $\G$, and the same is true for $\F$. 
So if $A_i$ and $A_i'$ are respectively the $i$-th distance matrices of $\G$ and $\F$, then there is a polynomial $p_i$ of degree $i$ such that
\[
    A'_i = p_i(A_1') = p_i(Q^\t A_1 Q) = Q^\t p_i(A_1) Q = Q^\t A_i Q
\]
for all $i$.

Suppose $\G$ and $\F$ are distance-regular graphs. 
Friedland \cite{friedland1989coherent} proved that if $\varPsi$ is an invertible coherent homomorphism between two coherent algebras, then there exists a unitary matrix $U$ such that $\varPsi(M) = U M U^{\ast}$ for all $M$. 
The converse is also true. 
In particular, there is an invertible coherent homomorphism $\varPsi$ between the coherent algebras generated by the distance matrices of $\G$ and $\F$ provided $\G$ and $\F$ are generalized cospectral.
By taking the restriction of $\varPsi$ to the coherent closures of $\G$ and $\F$, we obtain that $\G$ and $\F$ are WL-equivalent. Thus we obtain the following consequence of \lref{wl}. A proof of a stronger result will be given in Section \ref{sec:qpg}.

\begin{theorem}\label{thm:drg}
For distance-regular graphs, generalized cospectrality implies $C^3$-equivalence. 
\end{theorem}

Now assume $\G$ is a regular graph of degree $k$. Then we have that 
\[
    \ch{\comp{\G}}{x} = (-1)^n\ \frac{x - n + k + 1}{x + k + 1}\ \ch{\G}{-x-1}.
\]
Therefore, two regular graphs are cospectral if and only if they are generalized cospectral.
The unique distance-regular graph with intersection array $\{5,4,1,1;1,1,4,5\}$ is called the \emph{Wells graph}, see for instance \cite[Theorem 9.2.9]{brouwer1989}.
Van Dam and Haemers \cite{van2002spectral} showed that there is a regular graph cospectral with the Wells graph which is not distance-regular.
We have the following result which implies that the Wells graph and its cospectral mate are not $C^3$-equivalent.

\begin{theorem} \label{thm:intersection} 
For each intersection array $\cI$ there is a $C^3$-sentence $\phi_\cI$ such that a graph $\G$ satisfies $\phi_\cI$ if and only if $\G$ is distance-regular with intersection array $\cI$. 
\end{theorem}
\begin{proof}[\emph{Proof}]
Let $\cI =  \{b_0,b_1,\dots, b_{d-1};c_1,c_2,\dots,c_d\}$. 
We write a $C^3$-sentence which is true in a graph $\G$ if and only if $\G$ 
 is distance-regular with intersection array $\cI$. 
We can express the existence of a path from $x$ to $y$ of length $i\geq 0$ in the three-variable counting logic $C^3$ as follows: 
\[
\phi_0(x,y) : = (x = y), \quad\quad  \phi_1(x,y) : = E(x,y)
\]
and
\[
\phi_{i+1}(x,y) : = \exists z [E(x,z) \land \exists x (z=x \land \phi_i(x,y)) ].
\]
Let $\delta_i(x,y)$ be the $C^3$-formula defined by
\[
\delta_i(x,y):= \phi_i(x,y) \land \bigwedge_{j=0}^{i-1} \lnot \phi_{j}(x,y). 
\]
Hence $\G \models \delta_i(x,y)$ precisely when $x$ and $y$ are vertices at distance $i$ in $\G$. 
If $r$ is a positive integer, we define the $C^3$-sentences 
\[
   \gamma^{r}_i : = \forall x\forall y[\delta_{i}(x,y) \to \exists^{=r} z(E(y,z) \land \delta_{i-1}(x,z))]\ \textnormal{ for }\ i \geq 1.
\]
and
\[
   \beta^{r}_i : = \forall x\forall y[\delta_{i}(x,y) \to \exists^{=r} z(E(y,z) \land \delta_{i+1}(x,z))] \ \textnormal{ for }\ i \geq 0.
\]
Then $\G \models \gamma^{c_i}_i \land \beta^{b_i}_i$ if and only if for any two vertices $x$ and $y$ at distance $i$ in $\G$, there are exactly $c_i$ vertices adjacent to $y$ at distance $i - 1$ from $x$, and $b_i$ vertices adjacent to $y$ at distance $i + 1$ from $x$.
Therefore 
\[
\G \models \bigwedge_{i=1}^{d} (\beta^{b_{i}}_{i-1} \land \gamma^{c_i}_i)
\]
whenever $\G$ is distance-regular with intersection array 
$\cI$.   
\end{proof}

One consequence of the previous result is that all graphs which are $C^3$-equivalent to the Wells graph are isomorphic to it. 
Indeed
\[
\G \models \beta^{5}_0\land \beta^{4}_1 \land \beta^{1}_2\land \beta^{1}_3\land \gamma^{1}_1 \land \gamma^{1}_2 \land \gamma^{4}_3 \land \gamma^{5}_4
\]
if and only if $\G$ is distance-regular with intersection array $\{5,4,1,1; 1,1,4,5\}$, and the Wells graph is unique up to isomorphism with this intersection array. 

\begin{corollary}
There is a pair of nonisomorphic generalized cospectral regular graphs which
are not $C^3$-equivalent.
\end{corollary}
\begin{proof}[\emph{Proof}]
Let $\G$ be the Wells graph and let $\F$ be a nonisomorphic cospectral mate of $\G$. 
Then $\G$ is distance-regular and so it is regular. 
Also $\F$ is regular. 
Thus $\G$ and $\F$ cannot be $C^3$-equivalent for then we would have that $\F$ is distance-regular with the same intersection array as $\G$, and this would imply that $\G$ and $\F$ are isomorphic. 
\end{proof}

\section{Quotient-Polynomial Graphs}
\label{sec:qpg}

Quotient-polynomial graphs were introduced by Fiol and Penji\'c \cite{fiol2021symmetric} as a generalization of orbit-polynomial and distance-regular graphs.

Let $\G$ be a graph on $n$ vertices with adjacency matrix $A$. 
For each eigenvalue $\theta$ of $A$, let $E_{\theta}$ be the matrix representing the orthogonal projection onto the eigenspace belonging to $\theta$. 
Then $A$ has the spectral decomposition 
\[
    A  = \sum_{\theta} \theta E_{\theta}.
\]

For any two vertices $u$ and $v$ from $\G$, let $\mathbf{w}_\ell (u,v)$ be the number of walks of length $\ell$ between $u$ and $v$ in $\G$. 
Thus 
\[
    \mathbf{w}_\ell (u,v) = (A^\ell)_{u,v} = \sum_{\theta} \theta^\ell (E_{\theta})_{u,v}.
\]
Suppose $A$ has $d+1$ distinct eigenvalues  (and hence the diameter of $\G$ is at least $d$). 
We say that a partition $\pi = (C_1, \dots, C_m)$ of $V(\G) \times V(\G)$ is \emph{walk-regular} if all the pairs $(u,v)$ in cell $C_i$ have the same vector 
\[
\mathbf{w}(u,v)= (\mathbf{w}_0 (u,v), \dots, \mathbf{w}_{d} (u,v))^\t.
\] 
It follows that all pairs of vertices from the same cell are at the same distance.

Given a walk-regular partition $\pi = (C_1, \dots, C_m)$ of $\G$, let $B_i$ be the  
$n\by n$ 01-matrix with $uv$-entry equal to 1 if and only if $(u,v) \in C_i$. 
Let $w_{\ell}(i)$ be the the number of walks of length $\ell$ between any two vertices $u$ and $v$ such that $(u,v)\in C_i$. 
Then, for all $\ell = 0,\dots, d$, we have 
\[
    A^\ell = \sum_{i = 1}^m w_{\ell}(i)B_i. 
\]
Thus  $\cx[A]= \span\{A^0,\dots,A^d\} \subset \span\{B_1,\dots,B_m\}$, and so it is necessary that $m\geq d+1$.  
We say that $\G$ is \emph{quotient-polynomial} if there are polynomials $p_i$ with degree at most $d$ such that $p_i(A) = B_i$ for all $i$. 
It follows that $\G$ is quotient-polynomial if and only if $\cx[A]=\span\{B_1,\dots,B_m\}$ (and hence $m = d+1$). 
Every distance-regular graph is quotient-polynomial, see for instance \cite[Proposition 5.1]{fiol2016quotient}.

Fiol and Penji\'c \cite{fiol2021symmetric} showed that a graph is quotient-polynomial if and only if its adjacency algebra is closed under Schur product. 
Since $\pi$ is a partition, we have  $\sum_{i=1}^m B_i = J$ and so $J \in \cx[A]$.
It follows that $\G$ is quotient-polynomial  if and only if $\cx[A]$ is a coherent algebra. 
Therefore, a graph is quotient-polynomial if and only if its coherent closure is its adjacency algebra.

\begin{theorem}\label{thm:qpg}
Two quotient-polynomial graphs are generalized cospectral if and only if they are $C^3$-equivalent.
\end{theorem}
\begin{proof}[\emph{Proof}]
We show that for quotient-polynomial graphs generalized cospectrality implies $C^3$-equivalence. 
For the converse we already have \tref{c3}. 
Let $\G$ and $\F$ be graphs with adjacency matrices $A$ and $A'$, respectively. 
If $\G$ and $\F$ are generalized cospectral, then they are cospectral and so there exists an orthogonal matrix $Q$ such that $Q^\t A Q = A'$.
Let $B_i$ and $B_i'$ be the matrices associated with the walk-regular partitions of $\G$ and $\F$, respectively. 
Suppose $\G$ and $\F$ are quotient-polynomial. 
Then there is a polynomial $p_i$ of degree at most $d$ such that
\[
    B'_i = p_i(A') = p_i(Q^\t A Q) = Q^\t p_i(A) Q = Q^\t B_i Q
\]
for all $i$. 
Thus there is an invertible coherent homomorphism $\varPsi(M) = Q^\t M Q$ between the coherent closures of $\G$ and $\F$.
Therefore $\G$ and $\F$ are WL-equivalent, and hence \lref{wl} implies that $\G$ and $\F$ are $C^3$-equivalent.   
\end{proof}

Since quotient-polynomial graphs are a generalization of orbit-polynomial and distance-regular graphs, as a corollary of Theorem \ref{thm:qpg} we obtain Theorem \ref{thm:drg}.

\section{Controllable Graphs}
\label{sec:control}

Let $\G$ be a graph on $n$ vertices with adjacency matrix $A$, and let $\1$ denote the all-one vector of length $n$. 
The number of walks of length $\ell$ in $\G$ is equal to 
\[
    \tr(A^\ell J) = \1^\t A^\ell \1.
\]
Since 
\[
     \sum_{\ell \geq 0} A^\ell x^\ell = (I - xA)^{-1} 
\]
we have that the generating function for all walks in $\G$ is
\[
    \sum_{\ell \geq 0}\tr(A^\ell J) x^\ell = \1^\t (I - xA)^{-1} \1. 
\]
We say that two graphs are \emph{walk-equivalent} if their generating functions for all walks are equal. 
Since 
\[
    xI - (J -A) = (xI + A)(I-(xI+A)^{-1} J)
\] 
and 
\[
    I-(xI+A)^{-1} J = I-((xI+A)^{-1}\1 )\1^\t
\]
then using the identity $\det(\u\v^\t) = 1 - \v^\t \u$, we find that
\[
    \frac{\ch{\comp{\G}}{x-1}}{(-1)^{n}\ \ch{\G}{-x} } = 1 - \1^\t (xI + A)^{-1} \1  
\]
and so
\[
      \1^\t (I - xA)^{-1} \1 = \frac{1}{x} \left( \frac{\ch{\comp{\G}}{-x^{-1}-1}}{(-1)^{n}\ \ch{\G}{x^{-1}}} - 1 \right).
\]
Thus, if two graphs are generalized cospectral then they are walk-equivalent.

\begin{theorem} \label{thm:c2}
    $C^2$-equivalent graphs are walk-equivalent.
\end{theorem}
\begin{proof}[\emph{Proof}]
We write a $C^2$-formula $\psi^r_{\ell}(x)$ such that for any graph $\G$ and any vertex $u$ from $\G$, we have 
\[
    \G \models \psi^r_{\ell}(u)
\]
if and only if there are $r\geq 0$ walks of length $\ell \geq 1$ in $\G$ starting at $u$.
We proceed by induction on $\ell$.
If $\ell = 1$ then we let
\[
\psi^0_{1}(x) := \forall y\lnot E(x,y) \qquad \textnormal{and} \qquad 
\psi^r_{1}(x) := \exists^{=r} y\  E(x,y) \quad \textnormal{for}\ r>0.
\]
For $\ell > 1$, we define
\[
\psi^0_{\ell+1}(x) := \forall y ( E(x,y) \to \psi^0_{\ell}(y) )
\]
and if $r >0$ then 
\[
\psi^r_{\ell+1}(x) := \bigvee_{(r_1^{a_1},\dots,r_d^{a_d})\in \Pi_r} [ ( \bigwedge_{i = 1}^d \exists^{=a_i}y\ \psi_{\ell}^{r_i}(y)) \land \exists^{\geq a}y\ E(x,y)]
\]
where $\Pi_r$ is the set of integer partitions $(r_1^{a_1},\dots,r_d^{a_d})$ of 
$r$ (so $r= \sum_{i=1}^d a_ir_i$) and $a = \sum_{i=1}^d a_i$. 
With this notation we define the $C^2$-sentence $\phi_{\ell}^{r}$ as follows:
\[
    \phi_{\ell}^{r} := \bigvee_{(r_1^{a_1},\dots,r_d^{a_d})\in \Pi_r}\  \bigwedge_{i = 1}^d \exists^{=a_i}x\ \psi_{\ell}^{r_i}(x).
\]
By definition $\G \models  \phi_{\ell}^{r}$ if and only if there are $r$ walks of length $\ell$ in $\G$. 

Now suppose that $\G$ and $\F$ are two graphs which are not walk-equivalent. 
Then $\G$ and $\F$ have a different number of walks of length $\ell$ for some $\ell$. 
Let $r$ be the number of walks of length $\ell$ in $\G$. 
We have $\G \models  \phi_{\ell}^{r}$ and $\F \not\models  \phi_{\ell}^{r}$.
Therefore $\G \not\equiv_{C^2} \F$.
\end{proof}

\begin{corollary}
    Two regular graphs are walk-equivalent if and only if they are $C^2$-equivalent.
\end{corollary}
\begin{proof}[\emph{Proof}]
Let $\G$ be a $k$-regular graph on $n$ vertices with adjacency matrix $A$. 
 If $A$ has spectral decomposition 
 \[
        A = \sum_{\theta} \theta E_\theta
 \]
 then $A\1 = k\1$ and $E_k = n^{-1}\1\1^\t$. 
 Thus $E_\theta E_k =  E_\theta \1\1^\t  = 0$ and so $\1^\t E_\theta \1 = 0$ for all $\theta \neq k$. 
 Since 
 \[
(I - xA)^{-1} = \sum_{\theta} \frac{1}{1 - x\theta} E_\theta
 \]
we find that the generating function for all walks in $\G$ is
\[
    \1^\t (I - xA)^{-1}\1  =
    \frac{n}{1-xk}.  
\]
Hence,  any two $k$-regular graphs on $n$ vertices are walk-equivalent. 
Therefore if $\G$ and $\F$ are two regular graphs such that $\G \not\equiv_{C^2}\F$, then either the number of vertices in $\G$ is different than the number of vertices in $\F$ or the degree of $\G$ is different than the degree of $\F$.
In either case,  $\G$ and $\F$ are not walk-equivalent. 
The converse follows immediately from \tref{c2}.
\end{proof}

The \emph{walk matrix} of $\G$ is the $n\by n$ matrix 
\[
    W_\G  =  \big( \1\quad A\1\quad \cdots \quad A^{n-1}\1 \big).
\] 
Our next result combined with \tref{c3} and the fact that $C^3$-equivalent graphs are necessarily $C^2$-equivalent implies \cite[Lemma 3]{van2009developments}. 

\begin{lemma} \label{lem:walk}
    Let  $\G$ and $\F$ be $C^2$-equivalent graphs. 
    Then \[W^{\t}_{\G} W_{\G} = W^{\t}_{\F} W_{\F}.\]
\end{lemma}
\begin{proof}[\emph{Proof}]
We shall use the $C^2$-formulas $\psi^{r}_{\ell}(x)$ which are defined in the proof of \tref{c2}. 
Recall that $\psi^{r}_{\ell}(x)$ asserts the existence of exactly $r$ walks of length $\ell$ starting at $x$.
Since $\G$ and $\F$ are $C^2$-equivalent, 
there is a vertex $v$ from $\G$ such that $\G \models \psi^{r}_{\ell}(v)$ if and only if there is a vertex $v'$ from $\F$ such that $\F \models \psi^{r}_{\ell}(v')$. 
Hence the mapping $v \mapsto v'$ is a bijection between the sets
\[
\{v \in V(\G): \G \models\psi^{r}_{\ell}(v)\} \qquad \textnormal{and} \qquad \{v' \in V(\F): \F \models\psi^{r}_{\ell}(v')\}
\]
for each $r\geq 0$ and $\ell \geq 1$. 
Since the rows of the walk matrix of a graph are indexed by the vertices of the graph, it follows that $W_{\G}$ and $W_{\F}$ are equal up to a permutation of the rows induced by the above bijection. 
Thus there is a permutation matrix $P$ such that $P W_{\G} = W_{\F}$. 
Then 
\[
W^{\t}_{\F} W_{\F} = (P W_{\G})^{\t}(P W_{\G}) =  W_{\G}^{\t}P^{\t}P W_{\G} = W^{\t}_{\G} W_{\G}. \qedhere 
\]
\end{proof}

We say that a graph $\G$ is \emph{controllable} if its walk matrix $W_{\G}$ is invertible. 
The theory of controllable graphs was developed in \cite{godsil2012controllable}, where it was conjectured that the proportion of graphs on $n$ vertices which are controllable goes to 1 as $n \to \infty$. 
It was later confirmed in \cite{o2016conjecture} that indeed almost all graphs are controllable. 
Wang and Xu \cite{wang2006sufficient} proved that if $\G$ and $\F$ are two walk-equivalent controllable graphs with adjacency matrices $A$ and $A'$, respectively, then $Q = W_{\G}W_{\F}^{-1}$ is an orthogonal matrix such that $Q^{\t}AQ = A'$ and $Q\1 = \1$. 
Thus we have the following consequence of \tref{c2}, \lref{walk} and the preceding remark. 

 \begin{corollary} \label{cor:iso}
     Two controllable graphs are isomorphic if and only if they are $C^2$-equivalent.
 \end{corollary}
\begin{proof}[\emph{Proof}]
We prove that $C^2$-equivalent controllable graphs are isomorphic. Consider two controllable graphs $\G$ and $\F$ with adjacency matrices $A$ and $A'$, respectively. 
If $\G$ and $\F$ are $C^2$-equivalent, then by \tref{c2} $\G$ and $\F$ are walk-equivalent.
Hence the matrix $Q = W_{\G}W_{\F}^{-1}$ satisfies $Q^{\t}AQ = A'$.

Now, in the proof of \lref{walk} we saw that there is a permutation matrix $P$ such that $W_{\G} = P^\t W_{\F}$. 
Then we have 
\[
Q  = W_{\G}W_{\F}^{-1} = P^\t W_{\F}W_{\F}^{-1} = P^\t
\]
and thus $PAP^{\t} = A'$. 
Therefore $\G$ and $\F$ are isomorphic. 
\end{proof}

\begin{remark}
For any fragment $L$ of the first-order language of graphs, we say that a graph $\G$ is \emph{$L$-definable} if there is an $L$-sentence $\phi$ such that every graph satisfying $\phi$ is isomorphic to $\G$.
Since almost all graphs are controllable \cite{o2016conjecture}, \cref{iso} implies the classical result of Immerman and Lander \cite{immerman1990describing} that almost all graphs are $C^2$-definable. 
\end{remark}


\subsection*{Acknowledgments}
Part of this work was done during a Dagstuhl Seminar on Logic and Random Discrete Structures, Germany, in February 2022. The authors would like to thank the organizers of the workshop.

The research of A. Abiad is partially supported by the FWO grant number 1285921N. The research of O. Zapata is partially supported by the Sistema Nacional de In\-ves\-ti\-ga\-dores grant number 620178.

\bibliographystyle{acm}
\bibliography{main}

\begin{thebibliography}{10}

\bibitem{ananchuen2006cubic}
{\sc Ananchuen, W., and Caccetta, L.}
\newblock Cubic and quadruple {P}aley graphs with the $n$-e.c. property.
\newblock {\em Discrete Mathematics 306}, 22 (2006), 2954--2961.

\bibitem{blass1981paley}
{\sc Blass, A., Exoo, G., and Harary, F.}
\newblock Paley graphs satisfy all first-order adjacency axioms.
\newblock {\em Journal of Graph Theory 5}, 4 (1981), 435--439.

\bibitem{brouwer1989}
{\sc Brouwer, A.~E., Cohen, A.~M., and Neumaier, A.}
\newblock {\em Distance-Regular Graphs}.
\newblock Springer Berlin Heidelberg, 1989.

\bibitem{cai1992optimal}
{\sc Cai, J.-Y., F{\"u}rer, M., and Immerman, N.}
\newblock An optimal lower bound on the number of variables for graph
  identification.
\newblock {\em Combinatorica 12}, 4 (1992), 389--410.

\bibitem{chen2019lectures}
{\sc Chen, G., and Ponomarenko, I.}
\newblock Lectures on coherent configurations.
\newblock {\em
  \href{http://www.pdmi.ras.ru/~inp/ccNOTES.pdf}{http://www.pdmi.ras.ru/\~{}inp/ccNOTES.pdf}\/}
  (2019).

\bibitem{DSZ}
{\sc Dawar, A., Severini, S., and Zapata, O.}
\newblock Descriptive complexity of graph spectra.
\newblock {\em Annals of Pure and Applied Logic 170\/} (2019), 993--1007.

\bibitem{fiol2016quotient}
{\sc Fiol, M.~A.}
\newblock Quotient-polynomial graphs.
\newblock {\em Linear Algebra and its Applications 488\/} (2016), 363--376.

\bibitem{fiol2021symmetric}
{\sc Fiol, M.~A., and Penji{\'c}, S.}
\newblock On symmetric association schemes and associated quotient-polynomial
  graphs.
\newblock {\em Algebraic Combinatorics 4}, 6 (2021), 947--969.

\bibitem{friedland1989coherent}
{\sc Friedland, S.}
\newblock Coherent algebras and the graph isomorphism problem.
\newblock {\em Discrete Applied Mathematics 25}, 1-2 (1989), 73--98.

\bibitem{godsil1993algebraic}
{\sc Godsil, C.~D.}
\newblock {\em Algebraic Combinatorics}, vol.~6.
\newblock CRC Press, 1993.

\bibitem{godsil2012controllable}
{\sc Godsil, C.~D.}
\newblock Controllable subsets in graphs.
\newblock {\em Annals of Combinatorics 16}, 4 (2012), 733--744.

\bibitem{immerman1990describing}
{\sc Immerman, N., and Lander, E.}
\newblock Describing graphs: A first-order approach to graph canonization.
\newblock In {\em Complexity Theory Retrospective}. Springer, 1990, pp.~59--81.

\bibitem{JN}
{\sc Johnson, C.~R., and Newman, M.}
\newblock A note on cospectral graphs.
\newblock {\em Journal of Combinatorial Theory, Series B 28}, 1 (1980),
  96--103.

\bibitem{MLW}
{\sc Mao, L., Liu, F., and Wang, W.}
\newblock A new method for constructing graphs determined by their generalized
  spectrum.
\newblock {\em Linear Algebra and its Applications 447\/} (2015), 112--127.

\bibitem{o2016conjecture}
{\sc O'Rourke, S., and Touri, B.}
\newblock On a conjecture of {G}odsil concerning controllable random graphs.
\newblock {\em SIAM Journal on Control and Optimization 54}, 6 (2016),
  3347--3378.

\bibitem{van2002spectral}
{\sc Van~Dam, E.~R., and Haemers, W.~H.}
\newblock Spectral characterizations of some distance-regular graphs.
\newblock {\em Journal of Algebraic Combinatorics 15}, 2 (2002), 189--202.

\bibitem{van2009developments}
{\sc Van~Dam, E.~R., and Haemers, W.~H.}
\newblock Developments on spectral characterizations of graphs.
\newblock {\em Discrete Mathematics 309}, 3 (2009), 576--586.

\bibitem{W2013}
{\sc Wang, W.}
\newblock Generalized spectral characterization revisited.
\newblock {\em The Electronic Journal of Combinatorics 20\/} (2013), \#P4.

\bibitem{W2017}
{\sc Wang, W.}
\newblock A simple arithmetic criterion for graphs being determined by their
  generalized spectra.
\newblock {\em Journal of Combinatorial Theory, Series B 122\/} (2017),
  438--451.

\bibitem{wang2006sufficient}
{\sc Wang, W., and Xu, C.-X.}
\newblock A sufficient condition for a family of graphs being determined by
  their generalized spectra.
\newblock {\em European Journal of Combinatorics 27}, 6 (2006), 826--840.

\bibitem{WX}
{\sc Wang, W., and Xu, C.-X.}
\newblock On the asymptotic behavior of graphs determined by their generalized
  spectra.
\newblock {\em Discrete Mathematics 310\/} (2010), 70--76.

\bibitem{weisfeiler1968reduction}
{\sc Weisfeiler, B., and Leman, A.}
\newblock The reduction of a graph to canonical form and the algebra which
  appears therein.
\newblock {\em Nauchno-Technicheskaya Informatsia, Seriya 2}, 9 (1968), 12--16.

\end{thebibliography}
\end{document}